\documentclass[sigconf,nonacm]{acmart}
\AtBeginDocument{%
  \providecommand\BibTeX{{%
    \normalfont B\kern-0.5em{\scshape i\kern-0.25em b}\kern-0.8em\TeX}}}

\usepackage{xcolor,colortbl}
\theoremstyle{acmplain}
\newtheorem{theorem}{Theorem}[section]
\newtheorem{conjecture}[theorem]{Conjecture}
\newtheorem{proposition}[theorem]{Proposition}
\newtheorem{lemma}[theorem]{Lemma}
\newtheorem{corollary}[theorem]{Corollary}
\theoremstyle{acmdefinition}
\newtheorem{example}[theorem]{Example}
\newtheorem{definition}[theorem]{Definition}
\newtheorem{cexample}[theorem]{Counterexample}
\theoremstyle{remark}
\newtheorem{remark}[theorem]{Remark}

\newcommand{\y}{\boldsymbol{y}}
\renewcommand{\a}{\boldsymbol{a}}
\renewcommand{\b}{\boldsymbol{b}}
\renewcommand{\d}{\boldsymbol{d}}
\newcommand{\f}{\boldsymbol{f}}
\renewcommand{\i}{\boldsymbol{i}}
\renewcommand{\j}{\boldsymbol{j}}

\newcommand{\s}{\boldsymbol{s}}
\newcommand{\balpha}{\boldsymbol{\alpha}}
\renewcommand{\u}{\boldsymbol{u}}
\renewcommand{\v}{\boldsymbol{v}}
\newcommand{\e}{\boldsymbol{e}}
\newcommand{\x}{\boldsymbol{x}}
\newcommand{\bomega}{\boldsymbol{\omega}}
\newcommand{\brho}{\boldsymbol{\rho}}
\newcommand{\bdelta}{\boldsymbol{\delta}}
\newcommand{\bvarphi}{\boldsymbol{\varphi}}
\newcommand{\0}{\boldsymbol{0}}

\newcommand{\K}{\mathbb{K}}
\newcommand{\Kx}{\mathbb{K}[x]}
\newcommand{\Kxx}{\mathbb{K}[\boldsymbol{x}]}
\newcommand{\N}{\mathbb{N}}
\newcommand{\Z}{\mathbb{Z}}

\newcommand{\cA}{\mathcal{A}}
\newcommand{\cM}{\mathcal{M}}
\newcommand{\cN}{\mathcal{N}}
\newcommand{\cF}{\mathcal{F}}
\newcommand{\cP}{\mathcal{P}}
\newcommand{\cS}{\mathcal{S}}

\newcommand{\cK}{\mathcal{K}}
\newcommand{\bcF}{\overline{\cF}}
\newcommand{\bs}{\overline{s}}

\newcommand{\Mon}{\mathsf{Mon}}

\newcommand{\rank}{\text{rank}}
\newcommand{\rdeg}{\text{rdeg}}
\newcommand{\eps}{\boldsymbol{\varepsilon}}
\newcommand{\p}{\boldsymbol{p}}
\newcommand{\q}{\boldsymbol{q}}
\newcommand{\dr}{\mathsf{d}}
\newcommand{\bdr}{\boldsymbol{\dr}}
\newcommand{\pr}{\mathsf{p}}
\newcommand{\bpr}{\boldsymbol{\pr}}
\newcommand{\barK}{\overline{K}}
\newcommand{\babarK}{\widetilde{K}}

\setcopyright{acmcopyright}
\copyrightyear{2020}
\acmYear{2020}
\acmDOI{}
\acmConference[ISSAC 2020]{International Symposium on Symbolic and Algebraic Computation}{20-23 July 2020}{Kalamata, Messinia, Greece.}

\begin{document}

\setlength{\abovedisplayskip}{3pt}
\setlength{\belowdisplayskip}{3pt}

\title{On the Uniqueness of Simultaneous Rational Function Reconstruction}

\author{Eleonora Guerrini, Romain Lebreton, Ilaria Zappatore}
\email{guerrini, lebreton, zappatore@lirmm.fr}
\affiliation{%
  \institution{LIRMM, Universit\'e de Montpellier, CNRS}
  \city{Montpellier}
  \state{France}
}

\begin{abstract}
This paper focuses on the problem of reconstructing a vector of rational functions given some evaluations, or more generally given their remainders modulo different polynomials.
The special case of rational functions sharing the same denominator, \emph{a.k.a.}
Simultaneous Rational Function Reconstruction (SRFR), has many applications from linear system solving to coding theory, provided that SRFR has a unique solution.
The number of unknowns in SRFR is smaller than for  a general vector of rational function. This allows to reduce the number of evaluation points needed to guarantee the existence of a solution,  but we may lose its uniqueness. In this work, we prove that uniqueness is guaranteed for a generic instance. 
\end{abstract}

\maketitle
\section{Introduction}
The Vector rational function reconstruction (VRFR) is the problem of finding all rational functions $\v/\boldsymbol{d} = (v_1/d_1,\dots,v_n/d_n)$ which satisfy some degree constraints, given a certain number of their evaluations $(\v/\boldsymbol{d})(\alpha_j) = \bomega_j$. We consider a  generalized version of this problem, where we suppose to know the images modulo different polynomials $a_1,\ldots, a_n$, \emph{i.e.} $u_i = v_i/d \bmod a_i$ for $1\leq i \leq n$.
The \textit{Simultaneous Rational Function Reconstruction} (SRFR) problem is a
particular case of the vector rational function reconstruction where the rational functions $\v/d = (v_1/d,\dots,v_n/d)$ share the same denominator (see Section~\ref{sec: SRFR}). 
We can apply the SRFR in different problems: from the decoding of classic and interleaved Reed-Solomon codes to the polynomial linear system solving.
As in the classic rational function reconstruction we focus on the homogeneous linear system related to our equations in its weaker form, \emph{i.e.} $\v-d\u \equiv 0 \bmod \a$.
If the number of equations is equal to the number of unknowns minus one then there always exists a non-trivial solution.
From now on, we will assume to be in this case.
Note that the common denominator constraint of SRFR implies less unknowns than general VRFR, so less equations. 
This has a direct impact on the complexity of its applications.
However, the uniqueness in not anymore guarantied as shown in Counterexample~\ref{ex:counter}.
Having a unique solution is fundamental for decoding algorithms or Evaluation-Interpolation methods (like for instance in linear system solving). This paper focuses on the conditions that guarantee the uniqueness of solutions of the SRFR.


Previous works show that in the application of SRFR for polynomial linear system solving, the uniqueness is ensured under some specific degree conditions \cite{olesh_vector_2007}.
We have reasons to believe that we can generalize this result: we
conjecture that for almost all $(\v,d)$ the SRFR problem admits a unique solution (see Conjecture~\ref{conj:fg-gen}). 

We can learn more about the conditions of uniqueness from the results coming from error correcting codes. Interleaved Reed Solomon Codes (IRS) can be seen as the evaluation of a vector of polynomials $\v$. The problem of decoding IRS codes consists in the reconstruction of the vector of polynomials $\v$ by its evaluations, some possibly erroneous.
A classic approach to decode IRS codes is the application of the SRFR for instances $\u = \v + \e$ where $\e$ are the errors. Results from coding theory show that for all $\v$ and almost all errors $\e$, we get the uniqueness of SRFR for the corresponding instance $\u$ (provided that there are not too many errors) \cite{bleichenbacher_decoding_2003,brown_probabilistic_2004,schmidt_collaborative_2009}. There is a natural generalization of SRFR when errors occur (SRFRwE, see Section~\ref{sec:SRFRwE}), which can be seen as fractional generalization of IRS \cite{guerrini_polynomial_2019,glz_enhancing2020}. We conjecture that we can decode almost all codeword $(\v/d)$ and almost all errors $\e$ of this fractional code (Conjecture~\ref{conj:fge-gen}). In this paper we present a result which is a step towards this conjecture. We prove that uniqueness is guaranteed for a generic instance $\u$ of SRFR, (Theorem~\ref{thm:genUnicity}). Our result is valid not only given evaluations, but also in the general context of any moduli $a$.

Our approach to prove Theorem~\ref{thm:genUnicity} is to study the degrees of a relation module. Solutions of SRFR are related to generators of a row reduced basis of this $\Kx$-module which have a 
negative shifted-row degree. Shifts are necessary to integrate degree constraints. We show that for generic instances, there is only one generator with negative row degree, hence uniqueness of the SRFR solution.

Previous works studied generic degrees of different but related modules: \emph{e.g.}
for the module of generating polynomials of a scalar matrix sequence \cite{Vil97}, 
for the kernel module of a polynomial matrix and specific matrix dimensions \cite{jeannerod_essentially_2005}. Both cases does not consider any shift. The generic degrees also appear in dimensions of blocks in a shifted Hessenberg form. However, the link with the degree of a module is unclear and no
shift is discussed (shifted Hessenberg is not related to our shift) \cite{PS07}.
We prove our result for any shift and any matrix dimension by adapting some of their techniques, and by proving that they apply to the specific relation modules related to SRFR.
 
 In Section~\ref{sec:mot} we introduce the motivations of our work, started from the classic SRFR to the extended version with errors. We also show their respective applications in polynomial linear system solving and in error correcting algorithms.
 In Section~\ref{sec:prel}, we
 define the algebraic tools that we will use to prove our technical results of the Section~\ref{sec:relModule}.
In Section~\ref{sec:unique} we explain how these results are linked to the uniqueness of the solution of the SRFR and we finally prove the Theorem~\ref{thm:genUnicity} about the generic uniqueness.

\section{Motivations}\label{sec:mot}
\subsection{Rational Function Reconstruction}\label{sec: SRFR}
In this section we recall standard definitions and we state our problem, starting from rational function reconstruction and its application to linear algebra.
Let $\K$ be a field, $a, u\in \Kx$ with $\deg(u)<\deg(a)$. 
The \emph{Rational Function Reconstruction} (shortly RFR) is the problem of reconstructing a rational function $v/d \in \K(x)$ such that
\begin{equation}{\label{eq:RFR}}
\gcd(d,a)=1, \frac{v}{d}\equiv u \bmod a, \deg(v)<N, \deg(d)<D.
\end{equation}
We focus on the weaker equation:
\begin{equation}{\label{eq:WRFR}}
v \equiv du \bmod a, \deg(v)<N, \deg(d)<D.
\end{equation}
The RFR problem generalizes many problems including the \emph{Pad\'e approximation} if $a=x^f$ and the \emph{Cauchy interpolation} if $a=\prod_{i=1}^f(x-\alpha_i)$, where the $\alpha_i$ are pairwise distinct elements of the field $\K$.
The homogeneous linear system related to the Equation~\eqref{eq:WRFR} has
$\deg(a)$ equations and $N+D$ unknowns.
If $\deg(a) = N+D-1$, the dimension of the solution space of Eq.~\eqref{eq:WRFR} is
at least $1$ and it always admits a non-trivial solution. Moreover, such a
solution is unique in the sense that all solutions are polynomial multiples of a
unique one, $(v_{\min},d_{\min})$ (see \emph{e.g.\ }\cite[Theorem
5.16]{vzgathen_gerhard_2013}).  On the other hand, Equation~\eqref{eq:RFR} does
not always have a solution, but when a solution exists, it is unique. Indeed, it
is $v_{\min}/d_{\min}$ and we can reconstruct it by the \textit{Extended
  Euclidean Algorithm} (EEA). Throughout this paper, we will focus on
Equation~\eqref{eq:WRFR}.

The RFR can be naturally extended to the vector case as follows.
Let $a_1,\ldots, a_n\in \Kx$ with degrees $f_i=\deg(a_i)$ and $\u=(u_1,\ldots,u_n)\in \Kx^n$ where $\deg(u_i)<f_i$. Let $0 < N_i,D_i<f_i$. The \textit{Vector Rational Function Reconstruction} (VRFR) is the problem of reconstructing $(v_i, d_i)$ for $1\leq i \leq n$ such that
$
v_i \equiv d_iu_i \bmod a_i, \deg(v_i)<N_i, \deg(d_i)<D_i.
$
We can apply the RFR componentwise and so, if $f_i=N_i+D_i-1$, we can uniquely reconstruct the solution.
\begin{definition}{(SRFR)}
 Given $\u=(u_1,\ldots,u_n)\in \Kx^n$ where $\deg(u_i)<f_i$, and degree bounds $0< N_i<f_i$ and $0< D<\max_{1\leq i \leq n}\{f_i\}$, we want to reconstruct the tuple $(\v,d)=(v_1,\ldots,v_n,d)$ such that
\begin{equation}{\label{eq:SRFR}}
  v_i \equiv du_i \bmod a_i, \deg(v_i)<N_i, \deg(d)<D.  
\end{equation}
We denote $\cS_{\u}$ the set of solutions.
\end{definition}
The SRFR is then the problem of reconstructing a vector of rational functions with the same denominator. 
Therefore, if $f_i = N_i+D-1$ for $1\leq i \leq n$, we can uniquely reconstruct the solution. In this case, the common denominator property allows to reduce the number of unknowns, with an impact on the degree of the $a_i$'s. In detail, the number of equations of \eqref{eq:SRFR} is $\sum_{i=1}^nf_i$, while the number of the unknowns, \emph{i.e.\ }the coefficients of $\v$ and $d$, is $\sum_{i=1}^nN_i+D$. 
If 
\begin{equation}{\label{eq:existenceNumbPoints}}
\sum_{i=1}^nf_i=\sum_{i=1}^nN_i+D-1
\end{equation}
then Equation~\eqref{eq:SRFR} always admits a non-trivial solution. However, the uniqueness is not anymore guarantied.

\begin{cexample}
  \label{ex:counter}
  Let $\K = \mathbb{F}_{11}$, $n=2$, $N_1=N_2=2$, $D=3$ and
  $a_1=a_2=\prod_{i=1}^3(x-2^i)=x^3 + 8x^2 + x + 2$. Let $\u=\v/d$ with
  $\v=(2x + 6, 8x + 2)$ and $d=2x^2 + 2x + 2$ invertible modulo $a_i$. Then the
  SRFR with instance $\u$ has two $\Kx$-linearly independent solutions
  $(d,\v)=(4x^2 + 9x + 10, 0, 0)$ and $(d',\v')=(8x + 3, 9x + 5, 3x + 9)$.
\end{cexample}


Uniqueness is a central property for the applications of SRFR: unique decoding algorithms are essential in error correcting codes, and it is also a necessary condition to use evaluation interpolation techniques in computer algebra. 
The study of the bound on the number of equations which guaranties the uniqueness of SRFR has also repercussion on the complexity. Indeed, the complexity of decoding algorithms or evaluation interpolation techniques depends on this number of equations. So decreasing this number has a direct impact on the complexity.


We denote by $s$ the rank of the $\Kx$-module spanned by the solutions $\cS_{\u}$. Therefore, all solutions can be written as a linear combination $\sum_{i=1}^s  c_i p_i$ of $s$ polynomials $p_i$ with polynomial coefficients $c_i$.
The case $s=1$ corresponds to what we call uniqueness of the solution.
In \cite{olesh_vector_2007}, the authors studied the particular case where $a_1=\ldots=a_n=a$ and $N_1=\ldots=N_n=N$.
They proved the following,
\begin{theorem}{\cite[Theorem 4.2]{olesh_vector_2007}}
Let $k$ be minimal such that $\deg(a)\geq N+(D-1)/k$, then the rank $s$ of the solution space $\cS_{\u}$ satisfies
$s\leq k$.
\end{theorem}
Note that if $k=1$, the solution is always unique ($s=1$). This matches the uniqueness condition on the $\deg(a)$ of VRFR. On the other hand, if $k=n$ and $\deg(a)\geq N+(D-1)/n$ then $s\leq n$ which is always true. Hence in this case the theorem does not provide any new information about the solution space. 
This theorem represents a connection between the classic bound on the $\deg(a)= N+D-1$ which guaranties the uniqueness and the \emph{ideal} one, \emph{i.e.\ }$\deg(a)=N+(D-1)/n$ (see Equation~\eqref{eq:existenceNumbPoints}), which exploits the common denominator property.
They also proposed an algorithm that computes a complete basis of the solution space using $\mathcal{O}(nk^{\omega-1}B(\deg(a)))$ operations in $\K$ 
where 
 $2\leq \omega \leq 3$ is the exponent of the matrix multiplication and
$B(t):=M(t)\log t$ where $M$ is the classic polynomial multiplication arithmetic complexity (see ~\cite{vzgathen_gerhard_2013} for instance).
In \cite{nielsen_algorithms_2016} the complexity was improved. In particular, they introduced an algorithm that computes the solution space (in the general case of different moduli, \emph{i.e.\ }$a_1, \ldots, a_n$) with complexity $\mathcal{O}(n^{\omega-1}B(f)\log(f/n)^2)$ where $f=\max_{1\leq i \leq n}\{\deg(a_i)\}$.

We now came back to general case of the SRFR. The main result of this work is to prove that when the degree constraints guarantee the existence of the solution, then for almost all $\u$ we also get the uniqueness (see Theorem~\ref{thm:genUnicity}). 
\begin{theorem}{\label{thm:introourresult}}
  If Equation~\eqref{eq:existenceNumbPoints} is satisfied, then for almost all
  instances $\u$ the SRFR admits a unique solution, \emph{i.e.\ }it has rank $s=1$.
\end{theorem}{}
We will both use the expressions ``almost all'' or ``generic'', meaning that there exists a polynomial $R$ such that a certain property  is true for all instances that do not cancel $R$. In our case, we state that there exists a polynomial $R$ such that the SRFR admits a unique solution for all instances $\u$ such that $R(\u)\neq 0$.

The SRFR problem has a natural application in a linear algebra context. 
\paragraph{Application to polynomial linear system solving}
Suppose that we want to compute the solution of a full rank polynomial linear system, $\y(x)=A^{-1}\b\in \K(x)$ where $A\in \Kx^{n \times n}$ and $\b \in \Kx^{n \times 1}$, from its image modulo a polynomial $a(x)$. We will refer to this problem as \emph{polynomial linear system solving} (shortly PLS).
We remark that, by the Cramer's rule, $\y$ is vector of rational functions with the same denominator:
PLS is then a special case of SRFR.
In \cite{olesh_vector_2007}, the authors proved that 
the solution space is uniquely generated ($s=1$)
when $\deg(a)\geq N+(D-1)/n$ in the special case of $D=N=n\deg(A)$ and $\deg(A)=\deg(b)$. They exploited another bound on the degree of $a$ based on \cite{cabay_exact_1971}.


In view of Theorem~\ref{thm:introourresult} and as our experiments suggest, we
could hope for the following,
\begin{conjecture}\label{conj:fg-gen}
  If Equation~\eqref{eq:existenceNumbPoints} is satisfied then for almost all $(\v,d)$ with $\gcd(d,a_i)=1$, the SRFR with $\u=\frac{\v}{d}$ as input admits a unique solution.
\end{conjecture}
Since we have proved the uniqueness for generic instances $\u$, it would be sufficient to show the existence of an instance $\u$ of the form $\v/d$ to prove the conjecture.



\subsection{Reconstruction with Errors}
\label{sec:SRFRwE}
In this section we introduce the problem of the Simultaneous Rational Function with Errors (\cite{boyer_numerical_2014,kaltofen_early_2017,guerrini_polynomial_2019,pernet_high_2014,glz_enhancing2020}), \emph{i.e.\ }the SRFR in a scenario where errors may occur in some evaluations.
Throughout this section we suppose that $\K$ is a finite field of cardinality $q$, we fix $\balpha = \{\alpha_1,\ldots,\alpha_f\}$ pairwise distinct evaluation points in $\K$ and we consider the polynomial $a=\prod_{i=1}^f(x-\alpha_i)$.
\begin{definition}{(SRFR with Errors)}
  Fix $0< N,D,\varepsilon<f\leq q$.  An instance of the SRFR with errors
  (SRFRwE) is a matrix $\bomega\in \K^{n \times f}$ whose columns are
  $\bomega_j=\v(\alpha_j)/d(\alpha_j)+\e_{j}$ for some reduced
  $\v/d\in \K(x)^{n \times 1}$ and some error matrix $\e$. The reduced vector must
  satisfy $\deg(\v)<N$, $\deg(d)<D$ and $d(\alpha_i)\neq 0$. The error matrix
  must have its \emph{error support}
  $E:=\{1\leq j\leq f \mid \e_{j}\neq \boldsymbol{0}\}$ which satisfies
  $|E|\leq \varepsilon$.

  The solution of the SRFRwE instance $\bomega$ is $(\v,d)$.
\end{definition}{}
\paragraph{SRFRwE as Reed-Solomon code decoding}
We observe that if $n=1$ and $D=1$, $\v/d$ is a polynomial. Then the
SRFRwE is the problem of recovering a polynomial $v$ given evaluations, some of
which possibly erroneous. So in this case, SRFRwE is the problem of decoding an
instance of a \emph{Reed-Solomon code}.

Its vector generalization, that is $n>1$ and $D=1$, coincides with the decoding
of an \emph{homogeneous Interleaved Reed-Solomon (IRS) code}. Indeed, an IRS
codeword can be seen as the evaluation of a vector of polynomials $\v$ on
$\balpha$. Thus decoding IRS codes is the problem of recovering $\v$ from
$\bomega_j=\v(\alpha_j)+\boldsymbol{e}_j$.


Let us now detail how we can solve SRFRwE using SRFR. We use the same technique
of decoding RS and IRS codes
\cite{elwyn_r._berlekamp_error_1986,bleichenbacher_decoding_2003,puchinger_decoding_2017}. We
introduce the \emph{Error Locator Polynomial} $\Lambda=\prod_{j\in
  E}(x-\alpha_j)$. Its roots are the erroneous evaluations so
$\deg(\Lambda)=|E|\leq
\varepsilon$.  We consider the \emph{Lagrangian polynomials} $u_i\in
\K[x]$ such that $u_i(\alpha_j)=\omega_{ij}$ for any $1\leq i \leq
n$. The classic approach is to remark that $(\bvarphi,\psi)= (\Lambda(x)\v(x),
\Lambda(x)d(x))$ is a solution of
\begin{equation}{\label{eq:keyEqAsRFR}}
  \bvarphi=\psi\u \bmod \prod_{i=1}^f(x-\alpha_i).
\end{equation}
In order to reconstruct
$(\v,d)$ it suffices to study the set of
$(\boldsymbol{\varphi},\psi)$ which verify Equation~\eqref{eq:keyEqAsRFR} and
such that $\deg(\boldsymbol{\varphi})< N+\varepsilon$ and $\deg(\psi)<
D+\varepsilon$.  In this way we reduce SRFRwE to SRFR (see Eq.~\ref{eq:SRFR}).
Hence, if
$f=(N+\varepsilon)+(D+\varepsilon)-1=N+D+2\varepsilon-1$ we can uniquely
reconstruct every component of the vector
(\emph{cf.} \cite{boyer_numerical_2014,kaltofen_early_2017}).

It is possible to reduce the number of evaluations w.r.t. the maximal number of
errors $\varepsilon$ in the setting of IRS decoding ($D=1$).

\begin{theorem}[{\cite{bleichenbacher_decoding_2003,brown_probabilistic_2004,schmidt_collaborative_2009}}]
  \label{thm:kiayas}
  Fix $0< N,\varepsilon<f\leq q$ and $E$ such that $|E|\leq \varepsilon$. If
  $f=N-1+\varepsilon+\varepsilon/n$, then for all $(\v,1)$ and almost
  all error matrices $\e$ of support $E$, the SRFRwE admits a unique solution on
  the instance $\bomega$ where $\bomega_j=\v(\alpha_j)/d(\alpha_j)+\e_j$.  
\end{theorem}

We prove a similar result in the rational function case,
\begin{theorem}[{\cite{guerrini_polynomial_2019,glz_enhancing2020}}]
    \label{thm:ISIT19}
  Fix $0< N,D,\varepsilon<f\leq q$ and $E$ such that $|E|\leq \varepsilon$. If
  $f=N+D-1+\varepsilon+\varepsilon/n$, then for all $(\v,d)$ and almost
  all error matrices $\e$ of support $E$, the SRFRwE admits a unique solution on
  the instance $\bomega$ where $\bomega_j=\v(\alpha_j)/d(\alpha_j)+\e_j$.
\end{theorem}{}
Since the problem of SRFRwE reduces to a simultaneous rational function
reconstruction, the Equation~\eqref{eq:keyEqAsRFR} always admits a nontrivial
solution whenever $f=N+\varepsilon+(D+\varepsilon-1)/n$. Our ideal result would
be to prove a uniqueness result also in this case. Our experiments suggest the
following,
\begin{conjecture}\label{conj:fge-gen}
  Fix $0< N,D,\varepsilon<f\leq q$ and $E$ such that $|E|\leq \varepsilon$. If
  $f=N+\varepsilon+(D+\varepsilon-1)/n$, then for almost all $(\v,d)$ and almost
  all error matrices $\e$ of support $E$, the SRFRwE admits a unique solution on
  the instance $\bomega$ where $\bomega_j=\v(\alpha_j)/d(\alpha_j)+\e_j$.
\end{conjecture}

Note that Conjecture~\ref{conj:fg-gen} is for almost all fractions $(\v,d)$
whereas Theorems~\ref{thm:kiayas} and~\ref{thm:ISIT19} are for all fractions.
This difference is due to Counterexample~\ref{ex:counter}, which states that we
can not have uniqueness for all $(\v,d)$ when $f=N+(D-1)/n$. This latter number
of evaluations matches the one of Conjecture~\ref{conj:fg-gen} in the situation
without errors $\varepsilon=0$. Remark that this obstruction does not affect
Theorems~\ref{thm:kiayas} and~\ref{thm:ISIT19} because their number of
evaluations $f$ becomes $N+D-1$ when $\varepsilon=0$.

Our result Theorem~\ref{thm:introourresult} is a first step towards
Conjecture~\ref{conj:fg-gen}: Since uniqueness of the SRFR is true generic
instance $\bomega_j$, it remains to prove the existence of an instance of the
form $\v(\alpha_j)/d(\alpha_j)+\e_j$ for any $E$ such that
$|E| \leq \varepsilon$ to prove the conjecture.

The SRFRwE was first introduced by \cite{boyer_numerical_2014} in a special case of its application, \emph{i.e.\ }the Polynomial Linear System Solving with Errors, that we will introduce in the following paragraph.

\paragraph{Polynomial linear system solving with errors}
We now suppose that we want to compute the unique solution of a PLS
$\y(x)=\v(x)/d(x)=A^{-1}\b\in \Kx^{n \times n}$ in a scenario where some errors
occur
\cite{boyer_numerical_2014,kaltofen_early_2017,guerrini_polynomial_2019}. In
detail, we fix $f$ distinct evaluation points
$\balpha=\{\alpha_1,\ldots,\alpha_f\}$ such that $d(\alpha_i)\neq 0$. In our
model, we suppose that there is a black box which for any evaluation point
$\alpha_i$, gives a solution of the evaluated systems of linear equations,
\emph{i.e.\ }$\y_i=A(\alpha_i)^{-1}b(\alpha_i)$. However, this black box could do some errors in the computations. In
particular, an evaluation $\alpha_i$ is \emph{erroneous} if
$\y_i\neq \v(\alpha_i)/d(\alpha_i)$ and we denote by
$E:=\{i\mid \y_i\neq \v(\alpha_i)/d(\alpha_i)\}$ the set of erroneous
positions. We refer to the problem of reconstructing the solution of a PLS in
this model of errors as \emph{Polynomial Linear System Solving with Errors}
(shortly PLSwE).  We observe that if $i\in E$, then there exists a nonzero
$\boldsymbol{e}_i\in \K^{n\times f}$ such that
$\y_i =\v(\alpha_i)/d(\alpha_i)+\e_i$. Hence, this problem is a special case of
SRFRwE. Here we want to reconstruct a vector of rational functions which is a
solution of a polynomial linear system. Therefore, all the results about uniqueness
of the previous sections hold. Furthermore, in \cite{kaltofen_early_2017}
authors introduced another bound which guaranties the uniqueness based on the
bounds on the degree of the polynomial matrix $A$ and the vector $\b$.

\section{Preliminaries}\label{sec:prel}

In this section we will give some definitions and set out the notation that we
will use throughout this paper. We refer
to~\cite{neiger_bases_2016} for the definitions and lemmas of this section, and
for historical references.

\subsection{Row degrees of a $\Kx$-module}

Let $\K$ be a field and $\Kx$ the ring of polynomials over $\K$. We start by
defining the row degree of a vector, then of a matrix.  Let
$\p=(p_1,\ldots,p_\nu)\in \Kx^\nu = \Kx^{1\times \nu}$ and
$\s=(s_1,\ldots,s_\nu)\in \Z^{\nu}$ a shift. 

\begin{definition}[Shifted row degree]\label{def:rowdeg}
Let $r_i=\deg(p_i)+s_i$ for $1\leq i \leq \nu$.
The \textit{$\s$-row degree} of $\p$ is $\rdeg_{\s}(p)=\max_{1\leq i \leq \nu}(r_i)$.

We also denote $\p=([r_1]_{s_1},\ldots,[r_\nu]_{s_\nu})$ a vector of polynomials where $r_i=\deg(p_i)+s_i$.
\end{definition}
We can extend this definition to polynomial matrices. In fact, let $P\in \Kx^{\rho \times \nu}$ be a polynomial matrix, with $\rho\leq \nu$.
Let $P_{j,*}$ be the $j$-th row of $P$ for $1\leq j \leq \rho$.
We can define the \textit{$\s$-row degrees} of the matrix $P$ as $\rdeg_{\s}(P):=(r_1,\ldots, r_\rho)$ where $r_j:=\rdeg_{\s}(P_{j,*})$.

Let $\cN$ be a $\Kx$-submodule of $\Kx^\nu=\Kx^{1\times \nu}$.  Since $\Kx$ is a
principal ideal domain, $\cN$ is \textit{free} of \textit{rank}
$\rho:=\rank(\cN)$ less than $\nu$ \cite[Section 12.1, Theorem
4]{dummit_Abstract_2003}. Hence, we can consider a \textit{basis}
$P \in \Kx^{\rho \times \nu}$, \textit{i.e.\ }a full rank polynomial matrix, such
that
$\cN=\Kx^{1\times \rho}P=\{\boldsymbol{\lambda}P \mid \boldsymbol{\lambda}\in
\Kx^{1\times \rho}\}$.

Our goal is to define a notion of row degrees of $\cN$ in order to study later
the $\K$-vector space
$\cN_{< r} := \left\{\p \in \cN \ \middle| \ \rdeg_{\s}(\p) < r \right\}$
for some $r \in \N$. Different bases $P$ of $\cN$ have different $\s$-row
degrees so we need more definitions. We start with row reduced bases.

Let $\boldsymbol{t}=(t_1,\ldots,t_\nu) \in \Z^{\nu}$. We denote by $X^{\boldsymbol{t}}$ a diagonal matrix whose entries are $x^{t_1},\ldots,x^{t_{\nu}}$.

\begin{definition}[Shifted Leading Matrix]
 The $\s$-leading matrix of $P$ is a matrix in $\K^{\rho \times \nu}$, whose entries are the coefficient of degree zero of $X^{-\rdeg_{\s}(P)}PX^{\s}$.
\end{definition}{}
%
%
\begin{definition}(Row reduced basis)
A basis $P\in \Kx^{\rho \times \nu}$ of $\cN$ is \textit{$\s$-row reduced} (shortly $\s$-reduced) if its leading matrix $LM_{\s}(P)$ has full rank.
\end{definition}

This definition is equivalent to \cite[Definition 1.10]{neiger_bases_2016},
which implies that all $\s$-reduced basis of $\cN$ have the same row degree, up
to permutation. We now focus on the following crucial property.

\begin{proposition}{(Predictable degree property)}\label{prop:predictable}

 $P$ is $\s$-reduced if and only if for all $\boldsymbol{\lambda}=(\lambda_1,\ldots,\lambda_{\rho})\in \Kx^{1\times \rho}$,
$$
\rdeg_{\s}(\boldsymbol{\lambda}P)=\max_{1\leq i \leq \rho}(\deg(\lambda_i)+\rdeg_{\s}(P_{i,*}))=\rdeg_{\boldsymbol{d}}(\boldsymbol{\lambda})
$$
where $\boldsymbol{d}=\rdeg_{\s}(P)$.
\end{proposition}

The proof of this classic proposition can be found for instance in
\cite[Theorem 1.11]{neiger_bases_2016}. This latter proposition is useful because it implies that
$\dim_{\K} \cN_{< r} = \sum_{\{i | r_i < r\}}(r-r_i)$ where $(r_1,\dots,r_\rho)$
is the $\s$-row degree of any $\s$-reduced basis of $\cN$.

Since we will need to define the $\s$-row degrees of $\cN$ uniquely, not just
up to permutation, we need to introduce ordered weak Popov form, which relies on the notion of pivot.

\begin{definition}[Pivot]
Let $\p \in \Kx^{1 \times \nu}$. The \textit{$\s$-pivot index} of $\p$ is $\max \{ j \ | \ \rdeg_{\s}(\p)=\deg(p_j)+s_j\}$. Moreover the corresponding $p_j$ is the \textit{$\s$-pivot entry} and $\deg(p_j)$ is the \textit{$\s$-pivot degree} of $\p$. 
\end{definition}
We can naturally extend the notion of pivot to polynomial matrices.
%
\begin{definition}((Ordered) weak Popov form)
  \label{def:owp}
The basis $P$ of $\cN$ in \textit{$\s$-weak Popov form}  if the $\s$-pivot indices of its rows are pairwise distinct. On the other hand, it is in \textit{$\s$-ordered weak Popov form} if the sequence of the $\s$-pivot indices of its rows is strictly increasing.
\end{definition}{}

A basis in $\s$-weak Popov form is $\s$-reduced. Indeed, $LM_{\s}(P)$ becomes,
up to row permutation, a lower triangular matrix with non-zero entries on the
diagonal. Hence it is full-rank.

Assume from now on that $\cN$ is a submodule of $\Kx^\nu$ of rank $\nu$ and that
$P$ is a basis of $\cN$ in $\s$-ordered weak Popov form. Then its pivot indices
must be $\{1,\dots,\nu\}$.

Weak Popov bases have a strong degree
minimality property, stated in the following lemma. 

\begin{lemma}[{\cite[Lemma 1.17]{neiger_bases_2016}}]\label{lm:lemmPivot}
  Let $\s \in \Z^{\nu}$, $P$ be a basis of $\cN$ in $\s$-weak Popov form with
  $\s$-pivot degrees $(d_1, \ldots, d_{\nu})$. Let $\p \in \cN$ whose pivot
  index is $1 \leq i \leq \nu$.  Then the $\s$-pivot degree of $\p$ is
  $\geq d_i$ or equivalently $\rdeg_{\s}(\p) \geq \rdeg_{\s}(P_{i,*})$.
\end{lemma}{}

As it turns out, ordered weak Popov basis are reduced basis for which the $\s$-row degree is unique.
The following lemma is a
consequence of Lemma~\ref{lm:lemmPivot}.

\begin{lemma}[{\cite[Lemma 1.25]{neiger_bases_2016}}]
  \label{lm:uniqrdeg}
  Let $\s \in \Z^\nu$ and assume $\cN$ is a submodule of $\Kx^\nu$ of rank
  $\nu$.  Let $P$ and $Q$ be two bases of $\cN$ in $\s$-ordered weak Popov
  form. Then $P$ and $Q$ have the same $\s$-row degrees and $\s$-pivot degrees.
\end{lemma}


\subsection{Link between pivot and leading term}

In this section, we will focus on the relation between pivots of weak Popov
bases and leading terms w.r.t. a specific monomial order, as in Gr\"obner basis theory (see for instance
\cite{cox_using_1998}).

Let $\Kxx:=\K[x_1, \ldots, x_{n}]$ be the ring of multivariate polynomials.
Recall that a \textit{monomial in} $\Kxx$ is a product of powers of the
indeterminates $\x^{\i} := x^{i_1}_1\cdots x^{i_n}_{n}$ for some
$\i:=(i_1, \ldots, i_{n})\in \N^{n}$.  On the other hand, a \textit{monomial
  in} $\Kxx^{n}$ is $\x^{\i} \eps_j$, where $\eps_1, \ldots,\eps_{n}$ is the
canonical basis of the $\Kxx$-module $\Kxx^{n}$.

A \textit{monomial order} on $\Kxx^{n}$ is a \textit{total order} $\prec$ on
the monomials of $\Kxx^{n}$ such that, for any monomials
$\varphi \eps_{\i},\psi \eps_{\j}\in \Kxx^{n}$ and any monomial $\tau \neq 1$,
$\tau \in \Kxx$,
$$
\varphi \eps_{\i} \prec \psi \eps_{\j}
\Longrightarrow
\varphi \eps_{\i} \prec \tau\varphi\eps_{\i} \prec \tau \psi \eps_{\j}.
$$
Given a
monomial order $\prec$ on $\Kxx^{n}$ and $f\in \Kxx^{n}$, the
\textit{$\prec$-initial term} $in_{\prec}(f)$ of $f$ is the term of $f$ whose
monomial is the greatest with respect to the order $\prec$.
We remark that in the case of $\Kx$, the only monomial order must be the natural degree order $ x^a<x^b \Longleftrightarrow a<b$. 

\begin{definition}{(shifted-TOP order)}\label{def:shiftedTOP}
Let $\prec$ be a monomial order on $\Kxx$. We consider the $\Kxx$-module $\Kxx^{n}$ with its canonical basis $\eps_1, \ldots, \eps_{n}$ and let $\gamma_1, \ldots, \gamma_{n}$ be monomials in $\Kxx$.
Then $\prec$ induces the following monomial order on $\Kxx^{n}$ called $\s$-TOP (Term Over Position):
$$
\varphi\eps_i\prec_{\s-TOP} \psi\eps_j \iff
(\varphi \gamma_i \prec \psi \gamma_j)  \text{ or }  (\varphi \gamma_i = \psi \gamma_j \text{ and } i<j)$$
for any pairs of monomials $\varphi \eps_i$ and $\psi \eps_j$ of $\Kxx^{n}$.
\end{definition}{}

As for the univariate module $\Kx^{n}$, the only monomial order $\prec$ on $\Kx$ is the \textit{natural} one. The \textit{shifting monomials} are $x^{s_i}$, defined by the shift $\boldsymbol{s}=(s_1, \ldots, s_{n})\in \N^{n}$. Hence, the $\s$-TOP order on $\Kx^{n}$ is
\begin{equation}{\label{eq:STop}}
x^a\eps_i<_{\s\text{-TOP}}x^b\eps_j \Longleftrightarrow (a+s_i,i)\prec_{lex} (b+s_j,j)
\end{equation}
where $\prec_{lex}$ is the \textit{lexicographic order} on $\Z^2$.

We can now state the link between this monomial order and the pivot's definition:  let $\p\in \Kx^{1\times n}$ and $in_{\prec_{\s\text{-TOP}}}(\p)=\alpha x^d \eps_i$ be the $\prec_{\s-TOP}$-initial term of $\p$, then the $\boldsymbol{s}$-pivot index, entry, and degree are respectively $i$, $p_i$ and $d$. This will be useful later on, in \emph{e.g.\ }Proposition~\ref{prop:relation}.

\section{Row Degree of the Relation Module}\label{sec:relModule}
Fix $m\geq n\geq 0$, and $M \in \Kx^{m\times n}$.
We consider a $\Kx$-submodule  $\cM$ of $\Kx^n$. We define the $\Kx-$\textit{module homomorphism} 
$$
\begin{array}{cccc}
\hat{\varphi_M}: &\Kx^m &\longrightarrow &\Kx^n/\cM\\
&\p &\longmapsto &\p M 
\end{array}{}.
$$
Set $\cA_{\cM, M}:=\ker(\hat{\varphi_M)}$ to get the injection
$$
\varphi_M: \Kx^m/\cA_{\cM, M} \hookrightarrow \Kx^n/\cM.
$$
We call $\cA_{\cM, M}$ the \emph{relation module} because $p\in\cA_{\cM, M} \Leftrightarrow \varphi_M(\p)=\p M =  0 \bmod \cM$, \emph{i.e.\ }$\p$ is a relation between rows of $M$.

Let $\eps_1,\ldots, \eps_m$ be the \textit{canonical basis} of $\Kx^m$, $\eps'_1, \ldots, \eps'_n$ the \textit{canonical basis} of $\Kx^n$ and $\e_i \equiv \eps_i \bmod \Kx^m/A_{\cM, M}$ for $1\leq i \leq m$. 
\begin{remark}{\label{rem: invFactForm}}
We observe that by the \emph{Invariant Factor Form of modules over Principal Ideal Domains} (\emph{cf.}~\cite[Theorem 4, Chapter 12]{dummit_Abstract_2003}), 
$
\cK := \Kx^n/\cM 
  \simeq \Kx^n/ \left\langle a_i(x)\eps'_i \right\rangle_{1\leq i \leq n}
$
for nonzeros $a_i(x)\in \Kx$ such that 
$a_n(x)|a_{n-1}(x)|\ldots|a_1(x)$.
The polynomials $a_i(x)$ are the \emph{invariants} of the module $\cM$.
We also denote $f_i:=\deg(a_i(x))$ and we observe that
$f_1\geq f_2 \geq \ldots \geq f_n$.   
\end{remark}

From now on we will assume that
$\cM = \left\langle a_i(x)\eps'_i \right\rangle_{1\leq i \leq n}$. It means that
any $\q \in \cK$ can be seen as $(q_1 \bmod a_1, \dots, q_n \bmod a_n)$.
Using the result of Lemma~\ref{lm:uniqrdeg}, we can define the row and pivot
degrees of the relation module $\cA_{\cM,M}$.
\begin{definition}[Row and pivot degrees of the relation module]\label{def:rowPivrelationMod}
  Let $\s\in \Z^m$ be a shift and $P$ be any basis of $\cA_{\cM,M}$ in ordered
  weak Popov form. The $\s$-row degrees of the relation module $\cA_{\cM,M}$ are
  $\brho:=\rdeg_{\s}(P)=(\rho_1,\ldots,\rho_m)$ and the $\s$-pivot degrees are
  $\bdelta:=(\delta_1,\ldots, \delta_m)$ where $\delta_i=\rho_i-s_i$.

\end{definition}
Throughout this paper we will also denote $\brho_M$ and $\bdelta_M$ when we want
to stress out the matrix dependency.

\subsection{Row degree as row rank profile}

In this section, we will see that the row degrees of the relation module can be
deduced from the row rank profile of a matrix associated to $\hat{\varphi}_M$. We
start by associating the pivot degree of $\p \in \cA_{\cM, M}$ to linear dependency relation.
\begin{proposition}{\label{prop:relation}}
  There exists $\p\in \cA_{\cM, M}$ with $\s$-pivot index $i$ and $\s$-pivot
  degree $d$ if and only if $x^d\e_i\in B_M^{\prec x^d\eps_i}$ where
  $B_M^{\prec x^d\eps_i} := \langle x^n\e_j \mid x^n\eps_j
  \prec_{\s-TOP}x^d\eps_i\rangle$.
\end{proposition}
\begin{proof}
  Fix $i,d\in\N$ and let $\p\in \Kx^n$ with $\s$-pivot index $i$ and $\s$-pivot
  degree $d$, so $r:=\rdeg_{\s}(\p)=d+s_i$. Then
  $\p=([\leq r]_{s_1},\ldots, [\leq r]_{s_{i-1}},
  [r]_{s_i},[<r]_{s_{i+1}},\ldots,[<r]_{s_m})$ (see Definition~\ref{def:rowdeg}) and we can write $\p=cx^d\eps_i+\p'$ where $c\in \K^*$ and
  $\p'=([\leq r]_{s_1},\ldots, [\leq r]_{s_{i-1}},
  [<r]_{s_i},[<r]_{s_{i+1}},\ldots,[<r]_{s_m})$. So $\p \in \cA_{\cM, M}$ has
  $s$-pivot index $i$ and degree $d$ $\Leftrightarrow$
  $x^d\eps_i = -1/c \ \p' \bmod \cA_{\cM, M}$ $\Leftrightarrow$
\begin{equation*}
      x^d \e_i \in
      \left\langle x^n\e_j \ \middle| 
        \begin{array}{ll}
          n + s_j \leq d+s_i,&\text{for } 1\leq j\leq i-1\\
          n + s_j <    d+s_i,&\text{for } i\leq j\leq m 
        \end{array}
      \right\rangle = B_M^{\prec x^d\eps_i}. \ \square
\end{equation*}{}
\let\qed\relax
\end{proof}{}
\begin{theorem}{\label{thm:pivot}}
  Let $\bdelta$ be the $\s$-pivot degrees of the relation module $\cA_{\cM,M}$.
  Then $\delta_j=min\{d \mid x^{d}\e_j\in B_M^{\prec x^{d}\eps_j}\}$ for any
  $1\leq j\leq m$.
\end{theorem}{}
\begin{proof}
  Fix $1\leq j \leq m$. During this proof we denote
  $\overline{\delta}_j := min\{d \mid x^{d}\e_j\in B_M^{\prec x^{d}\eps_j}\}$. We want to prove that
  $\delta_j=\overline{\delta}_j$. Recall that by Proposition~\ref{prop:relation}, 
  $x^{\delta_j}\e_j\in B_M^{\prec x^{\delta_j}\eps_j}$. Hence, by the minimality
  of $\overline{\delta}_j$, $\delta_j\geq \overline{\delta}_j$. On the other
  hand,
  $x^{\overline{\delta}_j}\e_j\in B_M^{\prec x^{\overline{\delta}_j}\eps_j}$ so by Proposition~\ref{prop:relation}
  there exists $\p \in \cA_{\cM, M}$ of $\s$-pivot index $j$ and degree
  $\overline{\delta}_j$. Finally, by Lemma~\ref{lm:lemmPivot} we can conclude that
  $\overline{\delta}_j\geq \delta_j$.
\end{proof}{}

We now define the \emph{ordered matrix} $Mo_M$ as the matrix of
$\hat{\varphi}_M$ w.r.t. particular $\K$-vector space bases: the rows of $Mo_M$
from top to bottom are the monomials
of $\Kx^m$ sorted
increasingly for the $\prec_{\s-TOP}$ order (see Eq.~\eqref{eq:STop}). The
columns of $Mo_M$ are written \emph{w.r.t.} the basis
$\{\x^i\eps'_j\}_{\substack{1\leq j \leq n\\0\leq i < f_j}}$ of
$\Kx^n/\cM$. Therefore, $Mo_M$ has finite rank
$\rank(Mo_M) = \rank(\hat{\varphi}_M) = \rank(\varphi_M)$, infinite number of
rows and $(\sum_{i=1}^n f_i) = \dim_{\K}(\Kx^n/\cM)$ columns.


\paragraph{Monomial row rank profile}
Our goal is to relate the row rank profile of $Mo_M$ to the row degree of the
relation module. The classic definition of row rank profile of a rank
$r$ polynomial matrix is the lexicographically smallest sequence of $r$ indices
of linearly independent rows (\emph{cf.}~\cite{DPS15} for instance).  Since the
rows of our ordered matrix $Mo_M$ correspond to monomials, we will transpose the
previous definition to monomials instead of indices.

Let $\Mon_r$ be the sets of $r$ monomials of $\Kx^m$. We define the
lexicographical ordering on $\Mon_r$ by comparing lexicographically the sorted
monomials for $\prec_{\s-TOP}$. In detail, $ \cF <_{lex} \cF'$ iff
there exists $1 \leq t \leq r$ s.t. $x^{i_l} \eps_{j_l} = x^{u_l} \eps_{v_l}$ for $l < t$ and $x^{i_t} \eps_{j_t} \prec_{\s-TOP} x^{u_t} \eps_{v_t}$
where $\cF = \{x^{i_l} \eps_{j_l}\}_{1 \leq l \leq r}$ and
$\cF' = \{x^{u_l} \eps_{v_l}\}_{1 \leq l \leq r}$ and both
$\{x^{i_l} \eps_{j_l}\}$ and $\{x^{u_l} \eps_{v_l}\}$ are increasing for the
$\prec_{\s-TOP}$ order.

We will use this lexicographic order on monomials to define the row rank profile
of $Mo_M$.
Let $r=\rank(Mo_M)$.


\begin{definition}[Row rank profile]
  \label{def:rrp}
  For any matrix $M \in \Kx^{m \times n}$, we define the \emph{row rank profile} 
  of $Mo_M$ (shortly $RRP_M$) as the family of monomials of $\Kx^m$ defined by
  $RRP_M := min_{<_{lex}} \cP_M $ where  
  $$
  \cP_M := \left\{ \cF \in \Mon_r \ \middle| \ \{m M\}_{m \in \cF}
    \text{ are linearly independent in } \cK \right\}.
  $$
  
\end{definition}
We now introduce a particular family of monomials, that we will frequently use:
we will denote
$
\cF_{\d}:=\{x^i \eps_j\}_{\substack{i<d_j\\1\leq j \leq m}}
$
for any $\d = (d_1,\dots,d_m) \in \N^m$.

This family allows us to finally relate the row rank profile of $Mo_M$ to the row
degree of the relation module.

\begin{proposition}
  \label{prop:RRPandrdeg}
  The row rank profile of the ordered matrix $Mo_M$ is given by the pivot degrees $\bdelta_M$ of
  the relation module $\cA_{\cM,M}$, \emph{i.e.\ } $RRP_M = \cF_{\bdelta_M}$.
\end{proposition}
\begin{proof}
  We fix the matrix $M$ in order to simplify notations.  We define
  $\delta'_j = min \left\{ \delta \ | \ x^{\delta} \eps_j \notin RRP \right\}$
  and $\bdelta'=(\delta'_1, \dots, \delta'_m)$.  By properties of row rank
  profile, we have that $x^{\delta_j} \e_j \in B^{\prec x^{\delta_j} \eps_j}$
  (otherwise we could create a smaller family of linearly independent monomial
  with $x^{\delta_j} \e_j$). Using Theorem~\ref{thm:pivot}, we deduce that
  $\delta'_j \geq \delta_j$. Therefore
  $\cF_{\bdelta} \subset \cF_{\bdelta'} \subset RRP$. Since the families of
  monomials $\cF_{\bdelta}$ and $RRP$ have the same cardinality $r=\rank(Mo)$,
  they are equal so $\cF_{\bdelta} = RRP$.
\end{proof}

\subsection{Constraints on relation's row degree}

We will now focus on integer tuples $\bdelta_M$ which can be achieved.  For this
matter, in the light of Proposition~\ref{prop:RRPandrdeg}, we need to understand
which families $\cF_{\d}$ of monomials can be linearly independent in the
ordered matrix, \emph{i.e.\ }belong to $\cP_M$ (see Definition~\ref{def:rrp}).

Recall that
$\cK=\Kx^n/\cM=\Kx^n/\left\langle a_i(x)\eps'_i \right\rangle_{1\leq i \leq n}$
and $f_i=\deg(a_i(x))$ are non-increasing as in Remark~\ref{rem: invFactForm}. Recall also from Definition~\ref{def:rrp} that $\cP_M$ is the set of families $\cF$ of $r$ monomials in $\Kx^m$ such that $\{m M\}_{m \in \cF}$ are linearly independent in $\Kx^n/\cM$.

\begin{theorem}
  \label{thm:clement}  
  Let $\d \in \N^m$ be non-increasing. We can extend $\f \in \N^m$ by $f_{n+1}=\ldots=f_m=0$. Then 
  $\exists M\in \Kx^{m\times n}$ such that 
  $\cF_{\d} \in \cP_M$ 
  if and only if $\sum_{i=1}^l d_i\leq\sum_{i=1}^l f_i$ for all $1\leq l\leq m$.
\end{theorem}

The non-increasing property of $\d$ can be lifted: let $\d$ be non-increasing and $\d'$ be any permutation of $\d$. Then $\exists M\in \Kx^{m\times n}$ such that $\cF_{\d} \in \cP_M$ if and only if $\exists M'\in \Kx^{m\times n}$ such that $\cF_{\d'} \in \cP_{M'}$. Indeed, permuting $\d$ amounts to permuting the components of $\p$,\emph{i.e.\ }permuting the rows of $M$. This does not affect the existence property. 

The latter proposition is an adaptation of \cite[Proposition 6.1]{Vil97} and
its derivation \cite[Theorem 3]{PS07}. Even if the statements of these two
papers are in a different but related context, their proof can be applied almost
straightforwardly.
We will still provide the main steps of the proof, for the sake of clarity and also because we will have to adapt the proof later in Theorem~\ref{thm:genUnicity}.
Note also that we complete the 'if' part of the proof because it was not detailed in earlier references. For this matter, we introduce the following
\begin{lemma}
  \label{lm:dimrankr}
Let $\cN$ be a $\Kx$-submodule of $\cK$ of rank $l$.
Then the dimension of $\cN$ as $\K$-vector space is at most $f_1+\ldots+f_l$.
\end{lemma}

\begin{proof}
  First, remark that if $\boldsymbol{\q} \in \cN$ has its first non-zero element
  at index $p$ then $a_p(x) \boldsymbol{\q} = 0$. Now since $\cN$ has rank $l$,
  we can consider the matrix $B$ whose rows are the $l$ elements of a basis of
  $\cN$. We operate on the rows of $B$ to obtain the \emph{Hermite normal form}
  $B'$ of $B$. The rows $(\b'_i)_{1 \leq i \leq l}$ of $B'$ have first
  non-zero elements at distinct indices $k_1,\ldots,k_l$. Therefore
  $a_{k_j}(x) \b'_j = 0$ and
  $\{x^{i} \b'_j\}_{\substack{0\leq i <f_{k_j}\\1\leq j\leq l}}$ is a generating
  set of $\cN$ and so
  $\dim_{\K} \cN \leq f_{k_1}+\ldots +f_{k_l} \leq f_1+\ldots +f_l$ since
  $(f_i)$ are non increasing and $(k_j)$ pairwise distinct.
\end{proof}{}

\begin{corollary}{\label{cor:=>}}
  Let $r\geq 0$, $\d \in \N^l$ and $v_1,\ldots,v_l \in \cK $ such that
  $\{x^j\v_i\}_{\substack{0\leq j<d_i\\1 \leq i \leq l}}$ are linearly
  independent then $\sum_{i=1}^l d_i \leq \sum_{i=1}^l f_i$.
\end{corollary}

\begin{proof}
  We consider $\cN$ the $\Kx$-module spanned by $\{v_1, \ldots, v_l\}$, and we
  observe that $d_1+\ldots+d_l\leq \dim\cN\leq f_1+\ldots+f_l$ by
  Lemma~\ref{lm:dimrankr}.
\end{proof}{}

\begin{proof}[Proof of Theorem~\ref{thm:clement}]
  We observe that if $m > n$, we can write
  $\cK=\Kx^n/\left\langle a_i(x)\eps'_i \right\rangle_{1\leq i \leq n} =
  \Kx^m/\left\langle a_i(x)\eps_i \right\rangle_{1\leq i \leq m}$ where
  $a_j(x)=1$ for $n+1\leq j\leq m$. Hence we can suppose \emph{w.l.o.g.} that
  $m=n$.

$\Rightarrow)$ By the hypotheses, there exists a matrix $M \in \Kx^{m \times n}$ such that $\{x^i \eps_j M\}_{x^i \eps_j \in \cF_{\d}} = \{x^i \v_j\}_{0<i<d_j}$ are linearly independent in $\cK$ where $\v_j := \eps_j M$. 
Hence, for all $1 \leq l \leq m$, $\v_1,\dots,\v_l$ satisfy the conditions of the Corollary~\ref{cor:=>} and so $\sum_{i=1}^l d_i \leq \sum_{i=1}^l f_i$.

$\Leftarrow)$ Set $\u_i = \eps_i$ for $1 \leq i \leq m$ so that  $\{x^i\u_j\}_{\substack{i<f_j\\1\leq j\leq m}}$ are linearly independent in $\cM$.
We now consider the matrix
$
K:=[K_1|\ldots|K_m]
$
where $K_j\in \Kx^{m\times f_j}$ is in the \emph{Krylov} form,
that is 
$
K_j = K(\u_j,f_j) := [\u_j|x\u_j|\ldots|x^{f_j-1}\u_j]
$
by considering $\u_j$ as a column vector. Note that $K$ is full column rank by construction.
Our goal is to find vectors $\v_1,\dots,\v_m$ such that 
$[K(\v_1,d_1)|\dots|K(\v_m,d_m)]$
is full column rank (see $\babarK$ later).

For this matter, we first need to consider the matrix $\barK$ made of columns of
$K$ so that it remains full column rank. It is defined as
$ \barK:=[\barK_1|\ldots|\barK_m] $ where for $1\leq j \leq m$,
$\barK_j \in \Kx^{m\times d_j}$ are defined iteratively by
$$
\barK_j:=[K(\u_j,\min(f_j,d_j))|
K(x^{s_1}\u_{j_1}, t_1)|\ldots|K(x^{s_k}\u_{j_k},t_k)
]
$$
and $K(x^{s_l}\u_{j_l},t_l)$ derives from previously unused columns in $K$,
which we add from left to right, \emph{i.e.\ }$(j_l)$ are increasing. Since
$\sum_{i=1}^j d_i \leq \sum_{i=1}^j f_i$, we will only pick from previous
blocks, \emph{i.e.\ }$j_k<j$. Since we must have depleted a block $K_{i_l}$
before going to another one, we can observe that $s_l+t_l = f_l$ for $l<k$. The
last block $K_{i_k}$ is the only one that may not be exhausted, \emph{i.e.}
$s_k+t_k \leq f_k$. Conversely, $s_l=d_l$ for $l > 1$ because no columns have
been picked yet from the blocks $j_l$, except maybe the first block $j_1$ where
$s_1 \geq d_1$.

We want to transform $\barK_j$ into a Krylov matrix $\babarK_j$, working block by
block. First we extend $[K(\u_j,\min(f_j,d_j))|0|\dots|0]$ to the right to
$K(\u_j,d_j)$. Then we extend all blocks
$[0|\dots|0|K(x^{s_l}\u_{j_l}, t_l)|0|\dots|0]$ to the left and the right to
$K(x^{s'_l}\u_{j_l},d_l)$ where $s'_l$ equals $s_l$ minus the number of columns
of the left extension. In this way, the extension matches the original matrix on its
non-zero columns. Now we can define $\babarK:=[\babarK_1|\ldots|\babarK_m]$,
where $\babarK_j := K(\v_j,d_j)$ with
$\v_j := \u_j + \sum_{l=1}^k x^{s'_l} \u_{j_l}$.

A crucial point of the proof is to show that $s'_k \geq 0$. But since $d_i$
are-non increasing, $j_l$ are increasing and $j_k < j$, we get
$s_l \geq d_{j_l} \geq d_{j_k} \geq d_{j}$. As the number of columns of the left
extension is at most $d_j$, we can conclude $s'_k \geq 0$.

In \cite{Vil97} and \cite{PS07} it is proved that there exist an upper
triangular matrices $T$ such that $\babarK = \barK T$. So we can conclude that
$\babarK$, which is in the desired block Krylov form, is full column rank as is
$\barK$, which concludes the proof.
\end{proof}

\begin{example}
  We illustrate the construction of the proof of Theorem~\ref{thm:clement} with
  example. Let $m=4$, $n=3$, $\f=(8,4,4)$ extended to $f_4=0$ and
  $\d = (5,5,3,3)$. Remark that $\sum_{i=1}^l d_i\leq\sum_{i=1}^l f_i$ for all
  $1\leq l\leq m$. Then $\barK_1 = K(\u_1,d_1)$,
  $\barK_2 = [K(\u_2,f_2) | K(x^{d_1}\u_1,d_2-f_2)]$ picks its missing column
  from the first unused column of $K_1$, $\barK_3 = K(\u_3,d_3)$, and
  $\barK_4 = [K(\u_4,f_4)=\varnothing|K(x^{d_1+1}\u_1,f_1-(d_1+1)|
  K(x^{d_3}\u_3,f_3-d_3)]$ picks its 3 missing columns first from the 2 unused
  of $K_1$, then from the remaining one of $K_3$. 
  Then the construction extends $\barK$ to $\babarK = K(\v_i,d_i)$ where
  $\v_1=\u_1 = [1,0,0]$, $\v_2=\u_2+x^{d_2-(d_1-1)}\u_1=[x,1,0]$, $\v_3=\u_3=[0,0,1]$ and
  $\v_4=x^{d_1+1}\u_1 + x^{d_3-(f_1-(d_1+1))}\u_3=[x^6,0,x]$. Finally the matrix $M$ of
  the statement of Theorem~\ref{thm:clement} has its $j$-th row $M_{j,*}$ equal to $\v_j$.
  \hfill $\Diamond$
\end{example}

We now have all the cards in our hand to state the principal constraint on the pivot degree $\bdelta_M$ of the relation module $\cA_{\cM,M}$ when $M$ varies in the set of matrices $\Kx^{m \times n}$ such that $\rank(Mo_M)=\rank(\varphi_M)$ is fixed. We will denote by $\bdr_r$ the pivot degree corresponding to the constraint.

\begin{theorem}
  \label{thm:ineqrdeg}
  Recall that $\f=(f_1,\ldots,f_m)$ are the degrees of the invariants of $\cM$ where $f_i=0$ for $n+1\leq i \leq m$, and let
  $r = \rank(Mo_M)$. Then $\cF_{\bdelta_M} \geq_{lex} \cF_{\bdr_r}$ where
  \begin{equation}
    \label{eq:drg}
    \cF_{\bdr_r} = min_{<_{lex}} \left\{ \cF_{\d} \in \Mon_r
      \ \middle| \
      \forall 1 \leq l \leq m, \
      \sum_{i=1}^l d_i \leq \sum_{i=1}^l f_i
    \right\}
  \end{equation}
\end{theorem}

\begin{proof}
  We know from Proposition~\ref{prop:RRPandrdeg} that $RRP_M = \cF_{\bdelta_M}$
      so $\{x^i \eps_j M\}_{\substack{i < \delta_{j,M}\\1\leq j\leq m}}$ are linearly independent and $\sum_{i=1}^m \delta_{i,M} = r$. Using
  Theorem~\ref{thm:clement}, we get that
  $\sum_{i=1}^l \delta_{i,M} \leq \sum_{i=1}^l f_i$ for all $1 \leq l \leq m$. This means that $\cF_{\bdelta_M}$ belongs
  to the set whose minimum is $\cF_{\bdr_r}$, which implies our result.
\end{proof}

We observe that $r=\rank(Mo_M)$ must satisfy $0 \leq r \leq \Sigma:=\sum_{i=1}^m f_i = \dim_{\K} \Kx^n/\cM$ and that  $r = \Sigma$ is reachable since $m \geq n$.
Note also that $\bdr_r$ is well-defined in Theorem~\ref{thm:ineqrdeg} as long as $0 \leq r \leq \Sigma:=\sum_{i=1}^m f_i$ because it is related to the minimum of a non-empty set.

\subsection{Generic row degree of relation module}

We will now show that this pivot degree constraint $\bdr_{\Sigma}$ is attainable
by $\bdelta_M$ for matrices $M$ such that
$\rank(Mo_M)=\rank(\varphi_M)=\dim_{\K} \Kx^n/\cM$ in which case $\varphi_M$
becomes a bijection. More specifically, we will show that this is the case for
almost all matrices $M\in\Kx^{m \times n}$.

\begin{corollary}
  \label{cor:geneqrdeg}
  For a generic matrix $M \in \Kx^{m \times n}$, the pivot degrees $\bdelta_M$
  of the relation module $A_{\cM,M}$ satisfy $\bdelta_M = \bdr_{\Sigma}$ where
  $\Sigma=\sum_{i=1}^n f_i$.
\end{corollary}

\begin{proof}
  Since $\sum_{i=1}^l \dr_{\Sigma,i} \leq \sum_{i=1}^l f_i$ for all
  $1 \leq l \leq m$, we deduce from Theorem~\ref{thm:clement} that there exists
  $M \in \Kx^{m \times n}$ such that $\{ m M \}_{m \in \cF_{\bdr_{\Sigma}}}$ are
  linearly independent.  So the $\Sigma$-minor corresponding to those lines is
  non-zero for this matrix $M$.  We now consider this $\Sigma$-minor as a
  polynomial $R$ in the coefficients of $M$. This polynomial is then nonzero since
  it admits a nonzero evaluation.

  Now for any matrix $M=(m_{i,j})$ such that $R(m_{i,j}) \neq 0$, the vectors
  $\{ m M \}_{m \in \cF_{\bdr_{\Sigma}}}$ must be linearly independent, so
  $\rank(Mo_M)=\Sigma$. We have $RRP_M \leq_{lex} \cF_{\bdr_{\Sigma}}$ because
  $\cF_{\bdr_{\Sigma}} \in \cP_M$ (see Definition~\ref{def:rrp}).
  Theorem~\ref{thm:ineqrdeg} gives the other inequality, so
  $\cF_{\bdr_{\Sigma}}=RRP_M=\cF_{\bdelta_M}$ and $\bdelta_M = \bdr_{\Sigma}$.
\end{proof}

\subsubsection{Special cases}
In this section, we will see that our definition of the generic pivot degree
$\bdr_{\Sigma}$ in Eq.~\eqref{eq:drg} has a simplified expression in a wide
range of settings.
Set the notation $\bs = \max(\s)$. We will see that under some assumptions the
expected row degree $\bpr_{\Sigma} := \bdr_{\Sigma} + \s$ has a nice
form. Define $p$ and $u$ be the quotient and remainder of the Euclidean division
$\sum_{i=1}^m(f_i+s_i) = p\cdot m + u$.
The expected nice form of the row degrees will be
\begin{equation}
  \label{eq:nicepivotdeg}
  \p := (\underbrace{p+1,\dots,p+1}_{u \text{
      times}},\underbrace{p,\dots,p}_{m-u \text{ times}}).
\end{equation}
This nice form will appear the following conditions on $\f$ and
$\s$:
\begin{equation}
  \label{eq:hyp1}
  p \geq \bs 
\end{equation}
\begin{equation}
  \label{eq:hyp2}
  \forall 1\leq l \leq m-1, \
  \sum_{i=1}^l p_i \leq  \sum_{i=1}^l (f_i + s_i)
\end{equation}
\begin{theorem}
  \label{thm:drspecialcase}
  Let $\p$ as in Equation~\eqref{eq:nicepivotdeg}, ant let $\f$ be
  non-increasing such that Equations~\eqref{eq:hyp1} and~\eqref{eq:hyp2} hold.
  Then $\bpr_{\Sigma} = \p$.
\end{theorem}
This nice form of row degree was already observed in particular cases in
different but related settings. To the best of our knowledge, it can be found 
in \cite[Proposition 6.1]{Vil97} for row degrees of minimal generating matrix
polynomial but with no shift, in \cite[Corollary 1]{PS07} for dimensions of
blocks in a shifted Hessenberg form but the link to row degree is unclear and no
shift is discussed (shifted Hessenberg is not related to our shift $\s$), and in
\cite[after Eq. (2)]{jeannerod_essentially_2005} for kernel basis were $m=2n$
with no shifts.
\begin{proof}
  Denote again $\Sigma=\sum_{i=1}^n f_i$.  Let $\bcF$ be the first $\Sigma$
  monomials of $\Kx^m$ for the $\prec_{\s-TOP}$ ordering. Let
  $\p = (p+1,\dots,p+1,p,\dots,p)$ be the candidate row degrees as in the
  theorem statement and $\d = \p - \s$ be the corresponding pivot degrees. Note
  that Equation~\eqref{eq:hyp1} implies that $p \geq \bs$ so $\d \in \N^m$.

  First we show that Equation~\eqref{eq:hyp1} implies $\bcF = \cF_{\d}$.  For
  the first part, in order to prove $\bcF = \cF_{\d}$, we need to show that
  $d_i = \min\{d \in \N \ | \ x^d \eps_i \notin \bcF\}$. We already know that
  $d_i \in \N$. We will need to study the row degrees of the first monomials to
  conclude. The monomials of $\Kx^m$ of $\s$-row degree $r$ ordered increasingly
  for $\prec_{\s-TOP}$ are $[x^{r-s_i} \eps_i]$ for increasing $1 \leq i \leq m$
  such that $\ s_i \leq r$. There are $m$ such monomials when $r\geq \bs$. The
  monomials of $\s$-row degree less than $\bs$ are $\{x^i \eps_j\}_{i+s_j<\bs}$
  and their number is $\sum_{i=1}^m(\bs - s_i)$. From this we can deduce that
  the row degree of the $n$-th smallest monomial is
  $
  \left \lfloor (n - 1 - \sum_{i=1}^m (\bs-s_i))/m \right \rfloor + \bs
  = \left \lfloor (n - 1 + \sum_{i=1}^m s_i)/m \right  \rfloor
  $
  provided that $n \geq \sum_{i=1}^m (\bs-s_i) + 1$. We can now remark that the
  $(\Sigma+1)$-th smallest monomial has $\s$-row degree $p$. More precisely, the
  $(\Sigma+1)$-th smallest monomial is the $(u+1)$-th monomial of row-degree
  $r$, so $\bcF$ is equal to all monomials of row degree less than $p$ and the
  first $u$ monomials of row degree $p$. This proves
  $d_i = \min\{d \in \N \ | \ x^d \eps_i \notin \bcF\}$ and $\bcF = \cF_{\d}$.


  Second we deduce from Equation~\eqref{eq:hyp2} that for all $1 \leq l \leq m$,
  $\sum_{i=1}^l d_i = \sum_{i=1}^l (p_i-s_i) \leq \sum_{i=1}^l f_i$ , so
  $\cF_{\bdr_r} \leq_{lex} \cF_{\d}$ by Theorem~\ref{thm:ineqrdeg} and finally
  $\cF_{\bdr_r} = \cF_{\d}$ because $\bcF$ is the smallest set of $\Sigma$
  monomials.
\end{proof}

\begin{example}
  Here we provide 3 examples of generic row pivot $\bdr_{\Sigma}$ and row degree
  $\bpr_{\Sigma}$: Corollary~\ref{cor:geneqrdeg} applies only to the first
  situation because the second and third situations are made so that
  Eq.~\eqref{eq:hyp1} and respectively Eq.~\eqref{eq:hyp2} are not satisfied.
  Let $m=n=3$ and $\s=(0,2,4)$ so that $\bs = 4$ and $\sum(\bs-s_i)=6$.

  In the first situation $\f=(6,1,0)$, so $\sum(f_i+s_i)= 4*m+1$ and using
  Corollary~\ref{cor:geneqrdeg} we get $\bpr_{\Sigma}=(5,4,4)$ from
  Eq.~\eqref{eq:nicepivotdeg} and $\bdr_{\Sigma}=(5,2,0)$.  In the second
  situation, $\f=(3,0,0)$ and Eq.~\eqref{eq:hyp1} is not satisfied. We use
  Theorem~\ref{thm:drspecialcase} to get $\bdr_{\Sigma}=(3,0,0)$ from
  Eq.~\eqref{eq:drg} and $\bpr_{\Sigma}=(3,2,4)$.  Finally in the third
  situation, $\f=(3,3,1)$ and Eq.~\eqref{eq:hyp2} is not satisfied. We use
  Theorem~\ref{thm:drspecialcase} to get $\bdr_{\Sigma}=(3,3,1)$ from
  Eq.~\eqref{eq:drg} and $\bpr_{\Sigma}=(3,5,5)$.  Let $\cF_1,\cF_2,\cF_3$ be
  the respective families of monomial of the three situations. We picture these
  families in the following table, where $Mon$ are the first monomials for
  $\prec_{\s-TOP}$
  $$
  \begin{array}{r|ccccccccc}
    \hline
    Mon & \eps_1 & X\eps_1 & X^{2}\eps_1 & \eps_2 & X^{3}\eps_1 & X\eps_2 & X^{4}\eps_1 & X^{2}\eps_2 & \eps_3 \\
    \hline
    \rdeg_{\s} & 0 & \multicolumn{1}{|c}{1} & \multicolumn{2}{|c}{2} & \multicolumn{2}{|c}{3} & \multicolumn{3}{|c}{4} \\  
    \hline
    \hline
    \cF_1 & \bullet & \bullet & \bullet & \bullet & \bullet & \bullet & \bullet & & \\
    \hline
    \cF_2 & \bullet & \bullet & \bullet &   &   &   &   &   &  \\
    \hline
    \cF_3 & \bullet & \bullet & \bullet & \bullet &  & \bullet &   & \bullet  &  \bullet \\
    \hline
  \end{array}
  $$  
\end{example}


\section{Uniqueness Results on SRFR}\label{sec:unique}
Recall the SRFR, defined in Section~\ref{sec: SRFR}.
In particular, $a_1,\ldots,a_n\in \Kx$ with degrees $f_i:=\deg(a_i)$ and $\u:=(u_1,\ldots,u_n)\in \Kx^{n}$ such that $\deg(u_i) < f_i$ and $0< N_i\leq f_i$ for $1\leq i \leq n$, $0 < D\leq \min_{1\leq i \leq n}\{f_i\}$.
We want to reconstruct $(\v,d)=(v_1,\ldots,v_n,d)\in \Kx^{1\times (n+1)}$ such that
$
  v_i \equiv du_i \bmod a_i, \deg(v_i)<N_i, \deg(d)<D.
$

We consider $\cM=\langle a_i(x)\eps'_i\rangle$ and we denote by $S_{\u}$ the set of tuples which verify Eq.~\eqref{eq:SRFR}.
\begin{lemma}
  For the shift $\s=(-N_1,\ldots,-N_n, -D)\in \mathbb{Z}^{n+1}$, we have
  $(\v,d)\in S_{\u} \Leftrightarrow (\v,d)\in \cA_{\cM,R_{\u}}$ with
  $\rdeg_{\s}((\v,d))<0$, where
  \begin{equation}
    \label{eq:Ru}
    R_{\u} := \begin{bmatrix}
      \mathsf{Id}_n \\
      - \u
    \end{bmatrix} \in \Kx^{(n+1) \times n}
  \end{equation}
\end{lemma}
\begin{proof}
  Observe that $(\v,d)\in S_{\u}$ if and only if it satisfies the equation
  $\v - d \u \equiv (\v,d)R_{\u}\equiv 0 \bmod \cM$, that is
  $(\v,d) \in \cA_{\cM,R_{\u}}$, and if it satisfies the degree conditions
  equivalent to
  $\rdeg_{\s}((\v,d))=\max\{\deg(v_1)-N_1,\ldots,\deg(v_n)-N_n,\deg(d)-D\}<0$ (see
  Definition~\ref{def:rowdeg}).
\end{proof}{}{}

So in order to study the solutions of the SRFR we introduce the $\s$-row degrees
$\brho_{\u}:=\brho_{R_{\u}}$ and the $\s$-pivot indices
$\bdelta_{\u} := \bdelta_{R_{\u}}$ of $A_{R_{\u},\cM}$ (see
Definition~\ref{def:rowPivrelationMod}).
As remarked just after the \textit{predictable degree property}
(Proposition~\ref{prop:predictable}), 
\begin{equation}\label{eq:dimSu}
  \dim_{\K} S_{\u} = \dim_{\K} (A_{R_{\u},\cM})_{<0} = - \sum_{\rho_{\u,i} < 0} \rho_{\u,i}.
\end{equation}{}
We can now show our main theorem about uniqueness in SRFR for generic instances
$\u$.
\begin{theorem}\label{thm:genUnicity}
  Assume $\sum_{i=1}^{n}f_i=\sum_{i=1}^n N_i+D-1$. Then for generic
  $\u=(u_1,\ldots,u_n) \in \Kx^{1 \times n}$, the solution space $S_{\u}$ has
  dimension $1$ as $\K$-vector space.
\end{theorem}

\begin{proof}
  By the previous considerations (see Eq.~{\eqref{eq:dimSu}}) it is sufficient
  to prove that for generic $\u \in \Kx^{n+1}$,
  $\brho_{\u} = (0, \dots, 0,-1)$.

  First, we need to show that the generic $\s$-row degree $\bpr_{\Sigma}$ is the
  expected nice form $\p = (0, \dots, 0,-1)$ ($p=-1$ and $u=n=m-1$ because
  $\sum (f_j + s_j) = -1 \cdot m + (m-1)$, see Eq.~(\ref{eq:nicepivotdeg})).  It
  remains to check that we verify the hypotheses of
  Theorem~\ref{thm:drspecialcase}. By Equation~\eqref{eq:hyp1},
  $\bs \leq -1 = p$. By Equation~\eqref{eq:hyp2},
  $\sum_{i=1}^l p_i \leq 0 \leq \sum_{i=1}^l (f_i + s_i)$ for all
  $0 \leq l \leq m-1$ since $f_i + s_i \geq 0 \geq p_i$ for all $i$.
  
  It remains to show that there exists a matrix of the form $R_{\u}$ which
  satisfies the genericity condition of Corollary~\ref{cor:geneqrdeg}. Hence, the
  genericity condition is a non-zero polynomial when evaluated on matrices
  $R_{\u}$ and finally we have our result for generic $\u$.

  In order to do so, we show that the construction of the proof of the
  Theorem~\ref{thm:clement} provides a matrix of the form $R_{\u}$ in our case.
  In our case $(d_1,\ldots,d_{n+1})=(N_1,\ldots,N_n,D-1)$ and $m=n+1$, where
  $f_{n+1}=0$. In particular, by SRFR assumptions, for any $1 \leq i \leq n$,
  $d_i\leq f_i$ and so the matrices $\barK_i=[K(\u_i,d_i)]$ are already in the
  Krylov form. On the other hand, the last matrix is in the form
  $\barK_{n+1} = [K(x^{d_j} \u_j,t_j)]_{1 \leq j \leq n}$ where
  $d_j+t_j=f_j$. Then $\babarK_{n+1} = [K(\sum_{j=1}^n x^{s'_j} \u_j,d_j)]$ and
  we need to prove that $s'_j \geq 0$ differently because we don't have the
  assumption about the non-increasing $\d$. Recall that $s'_j$ is $s_j$ minus the number
  of columns added to extend the matrix to the left. This number of columns is
  at most $d_{n+1}$ minus the size $t_l$ of the current block. So
  $s'_l \geq d_l- (d_{n+1} - t_l) = d_l - (d_{n+1} - (f_l-d_l)) = f_l - d_{n+1}
  \geq 0$ because $d_{n+1} = D-1 \leq D \leq \min(f_i) $ and so the construction
  works.
\end{proof}



\end{document}